\documentclass[a4paper,twocolumn,11pt]{quantumarticle}
\pdfoutput=1
\usepackage{amssymb}
\usepackage{amsthm}
\usepackage{mathtools}
\usepackage{dsfont}
\usepackage{graphicx}
\usepackage{hyperref}
\newcommand{\tr}{{\rm tr}\,}
\newcommand{\ket}[1]{\left|{#1}\right\rangle}
\newcommand{\bra}[1]{\left\langle{#1}\right|}

\newcommand{\ketbra}[2]{\left|{#1}\rangle\!\langle{#2}\right|}

\DeclareMathOperator{\Tr}{Tr}

\newtheorem{theorem}{Theorem}
\newtheorem{definition}{Definition}
\newtheorem{corollary}{Corollary}
\newtheorem{postulate}{Postulate}

\begin{document}
\title{Classicality from Quantum Stochastic Processes}
\date{November 24 2023}
\author{Esteban Mart\'inez Vargas}
\email{Esteban.Martinez@uab.cat}
\affiliation{F\'isica Te\`orica: Informaci\'o i Fen\`omens Qu\`antics, Departament de F\'isica, Universitat Aut\`onoma de Barcelona, 08193 Bellatera (Barcelona) Spain}

\begin{abstract}
    I develop a theory of classicality from quantum systems. This theory stems
    from the study of classical and quantum stationary stochastic processes.
    The stochastic processes are characterized by polyhedral (classical) and
    semidefinite representative (quantum) cones.
    Based on a previous result by the author I expand the study of 
    fixed points from quantum channels. I give a semidefinite program that
    characterizes a quantum channel separating into a core and a part that
    decays with many iterations. In general, the solution is non-separable in the
    space it is defined. I present a characterization of channels in terms of 
    their fixed points for the separable case. A quantum simulation of a polyhedral 
    cone can then be constructed.
\end{abstract}

\maketitle
\section{Introduction}
Stochastic processes is a very
general framework to describe a wide variety of systems, in biology, economics,
chemistry, physics, etc. \cite{van1992stochastic,mitzenmacher_upfal_2005}. 
Specifically, the modeling through Hidden Markov Models has been widely studied \cite{ATutorialOnHRabine1989}.
These kinds of models arise when we have a stationary structure. 
These objects have 
been widely studied classically and in quantum systems \cite{QuantumStochasMilz2021}. 
Quantum stochastic processes are needed because quantum systems are always open to complex environments that
affect the evolution of a system and for foundational aspects of quantum mechanics \cite{breuer2002theory}.
Even though intuition indicates that the simplest and most practical description would 
correspond to classical information theory, several results show that
there is an advantageous simplicity when using quantum systems to describe stochastic 
processes even when having classical systems \cite{QuantumMechaniGuMi2012,DimensionalQuaGhafar2019,QuantumAdvantaKorzek2021}.
This means that quantum mechanics would be a natural language for stochastic processes.
It reminds of Ockham's razor \cite{audi_2015}, the philosophical principle that states 
that when describing the cause of a phenomenon a good practice is to think in the
simplest possible explanation. 
Nevertheless, this perspective would imply a paradoxical worldview: if quantum mechanics
is the most natural way to describe stochastic processes, which is a very general tool
to describe aspects of the world (classical and quantum) then why does our world seem classical?

Quantum stochastic processes thus remit us to a question that was 
raised since the inception of quantum mechanics a hundred years ago, that is
the passage from quantum to classical dynamics is usually expressed as the
correspondence principle \cite{BOHR1928}.
There is a large amount of work in this respect, specifically in the area of einselection 
and quantum Darwinism
\cite{DecoherenceEiZurek2003,QuantumDarwiniZurek2009,QuantumDarwiniRiedel2010}.
Existing an advantage in terms of information one would ask why would classicality even exist. 
Here we aim to study this topic using the formalism of quantum stochastic processes. 

The approach that the theory of einselection assumes is that there is a constant
feature of a system, which is its environment (or system plus environment). It thus aims to study the features of the
system which are resistant to decoherence. Classicality from a quantum perspective
is thus defined as those features which persist in time.
In this line of thought, Hidden Markov Models arise when the notion of stationarity
is relevant, therefore, such stationarity of a stochastic process remits us to think in a persistent structure that produces it.

However, there exist stationary processes produced from quantum sources, as 
characterized by Monràs and Winter in their ('O Scarrafone) theorem \cite{QuantumLearninMonras2016}.
There, they characterize the most general stationary stochastic process produced
by a quantum source. Therefore, stationarity in itself is not sufficient for 
the notion of classicality. In einselection a central aspect is also the objectivity
of a specific basis, there are einselected states. The fact that these states are
unchanged by the dynamics is also central to the objectivity of the system.
Therefore, we will understand classicality in a quantum stochastic process if it 
fulfills two conditions:
\begin{itemize}
    \item Persistence in time.
    \item Objectivity of its generating states.
\end{itemize}

In this paper, we explore a possible mechanism for the persistence in
time of an objective set of quantum states that gives place to a Markov process. 
We consider discrete uses of a quantum channel and the objective set will be made
of the fixed points of the channel. 
Observe that quantum channels have at 
least one fixed point \cite{WatrousTheTheoryof2018}. 
We here consider channels with multiple fixed points. A finite collection of such
points which are vectors form what is known as a polyhedral cone. Following a theorem
by Dharmadhikari \cite{SufficientCondSWD1963} this is a necessary and sufficient 
condition to have a stationary Markov process. We thus restrict Monràs and Winter
conditions to those asked by Dharmadhikari.

Although the fixed points of quantum channels have been studied in the literature \cite{WatrousTheTheoryof2018,FixedPointsOfArias2002}, 
it is a cumbersome topic as a part of some theorems, given a channel there is no general
way of finding its fixed points and
numerical simulations are almost always the norm \cite{OnPeriodCyclCarbon2020}. 
Here we give a characterization of quantum channels inversely:
given a finite number of states explore all the quantum channels that have them as
fixed points. We extend the results of \cite{2209.06806v1} to consider multiple fixed point
channels and give an example of the power of this approach. 

First, we introduce Dharmadhikari's and Mondràs and Winter theorems. Then explain the problem of classicality from the quantum perspective as a ``cone reduction'' problem. 
Finally, we introduce our study of multiple fixed point quantum channels and apply it to
an example. We finish with a discussion.

\section{Quasi-realizations of stochastic processes}
To study the different dynamics, classical and quantum, a general framework of
stationary stochastic processes will be introduced: the theory of quasi-realizations. 
From an abstract point of view, a stochastic process is given by an alphabet of symbols
$\mathcal{M}$ with size $|\mathcal{M}|=m$ and we denote $\mathcal{M}^l$ the set of words
of length $l$. We define
\begin{equation}
    \mathcal{M}^* = \cup_{l\geq0}\mathcal{M}^l.
\end{equation}
We can obtain the probability of a specific word $\mathbf{u}=(u_1,u_2,\ldots,u_l)\in\mathcal{M}^l$,
\begin{equation}
    p(\mathbf{u})=p(Y_1=u_1,Y_2=u_2,\ldots,Y_n=u_n).
\end{equation}
It will be relevant to study stationary distributions as they describe the stochastic process 
asymptotically.

We would like to infer what is the inner mechanism that 
gives rise to a stochastic process $Y$. A very widely known kind of matrices that produce a stochastic process are the stochastic matrices, which are non-negative
matrices whose rows sum 1. However, in general, the hidden mechanism of a 
stochastic process need not be described by a stochastic matrix. 
To see clearly this affirmation we make the following definition  
\begin{definition}
    A quasi-realization of a stochastic process is a quadruple ($\mathcal{V},\pi,D,\tau$),
    where $\mathcal{V}$ is a vector space, $\tau\in\mathcal{V}$, $\pi\in\mathcal{V}^*$, the dual
    space to $\mathcal{V}$ and
    $D$ is a unital representation of $\mathcal{M}^*$ over $\mathcal{V}$,
    \begin{align}
        D^{(\varepsilon)} &= \mathds{1},\\
        D^{(u)}D^{(v)}&=D^{(uv)},\quad\forall u,v\in\mathcal{M}^*.
    \end{align}
\end{definition}

\begin{definition}
    We will call \textbf{cause matrices} to the matrices 
    \begin{equation}
        D^c = \sum_{u\in\mathcal{M}}D^{(u)},
    \end{equation}
    where $D^{(u)}$ was defined above.
\end{definition}

Observe that several quasi-realizations can yield the same stochastic process, we will call
\emph{equivalent} two quasi-realizations that generate the same stochastic process. 
For us, the relevant outcome of this definition will be the quantum version, which means, to find 
stochastic processes where their cause matrices are not necessarily stochastic matrices.
\subsection{Classical cones: Dharmadhikari's theorem}
Observe now that being the stochastic matrices a subclass of cause matrices then
we need to find out the conditions when the cause matrices of a quasi-realization
($\mathds{R}^d,\pi,M,\tau$) become nonnegative matrices and
\begin{equation}
    M^s = \sum_{u\in \mathcal{M}}M^{(u)},
    \label{eq:BigM}
\end{equation}
is stochastic. 
This is precisely given by Dharmadhikari's theorem, it gives us the conditions
for having a positive realization, which means, that eq. (\ref{eq:BigM}) 
is fulfilled, $\pi\in(\mathds{R}^d)^*$ is a stationary distriburion and $\tau=(1,1,\ldots,1)$. 

\begin{theorem}
    \label{thm:dharma}
    Given a quasi-realization ($\mathcal{V},\pi,D,\tau$), an equivalent positive
    realization exists if and only if there is a convex pointed polyhedral cone
    $\mathcal{C}\subset\mathcal{V}$ such that
    \begin{itemize}
        \item $\tau\in\mathcal{C}$.
        \item $D^{(v)}(\mathcal{C})\subseteq\mathcal{C}$.
        \item $\pi\in\mathcal{C}^*$.
    \end{itemize}
    With $\mathcal{C}^*$ the dual cone of $\mathcal{C}$.
\end{theorem}
We thus need that all the dynamics to be restricted to a polyhedral cone.
\subsection{SDR cones: Scarrafone}
For quantum systems the type of cause matrices is different. The 
characterization of the quasi-realizations, in this case, is related to the
theorem in ref \cite{QuantumLearninMonras2016}.
\begin{theorem}
    \label{thm:scarrafone}
    Given a quasi-realization $(\mathcal{V},\pi,D,\tau)$ an
    equivalent, finite-dimensional, unital, completely positive realization
    $(\mathcal{B(H)}^{sa},\rho,\varepsilon,\mathcal{I})$ exists if and only
    if there is a SDR cone $\mathcal{P}\subset\mathcal{V}\otimes\mathcal{V}^*$
    such that
    \begin{itemize}
        \item $\tau\in\mathcal{C}$.
        \item $D^{(u)}\in\mathcal{P}$ for all $u\in\mathcal{M}$.
        \item $\pi\in\mathcal{C}^*$.
    \end{itemize}
    With $\mathcal{C}$ a SDR cone and $\mathcal{C}^*$ its dual defined in
    \cite{QuantumLearninMonras2016}. 
\end{theorem}
\section{Classical dynamics as a fixed point problem}
Any quantum dynamics is described by an SDR cone. This cone includes an instrument that 
maps states from the cone into itself. Theorem (\ref{thm:scarrafone}) offers a description
for general instruments. Observe, however, that for instruments with a constant channel 
in all iterations asymptotic behavior of channels becomes relevant. 
This situation implies analyzing the channels from the perspective of their fixed points,
as all channels have at least one by theorem (4.24) of \cite{WatrousTheTheoryof2018}.

The theory of quantum einselection by Zurek et. al. implies a model of reduction from quantum dynamics 
to classical dynamics. It states that the classical behavior of a system is given by certain states
that are selected over time by the interaction Hamiltonian which means, the interaction with 
an environment. Observe that this approach implies a continuous evolution of a system. The einselected
states are the ones that remain over time. Another important aspect is that the environment, and thus
the Hamiltonian remains constant over time. A Hamiltonian of changing potential would change itself 
and therefore einselection would not be possible.

We here want to study an analogous point of view. The classical dynamics will be given by a 
polyhedral cone. The reduction from quantum to classical will be given in terms of fixed points of
a quantum channel. The time evolution will thus be discrete, but always maintaining a channel constant.
We do not require that an environment imposes some evolution but we require constant dynamics. 
The classical cone is a result of constant quantum dynamics.

This point of view allows us to extend the einselection process from only a set of pure states to
possibly mixed states. Such states would not necessarily be orthogonal. We will show a decomposition
theorem for channels, to decompose any channel into its fixed point plus a part that goes to zero.
Then we show a bound to obtain this decomposition given a channel. Finally, we explore the decomposition
of multiple fixed points, which would imply a polyhedral cone.
\section{Cone reduction problem}
In this section, we develop a general theory to describe the transition between
quantum mechanics and classical mechanics. 
Starting from the historical perspective, the theorem (\ref{thm:dharma}) by Dharmadhikari finds the
conditions for the existence of equivalent positive realizations given a realization.
We therefore can translate this mathematical statement into a physically relevant
one with the following postulate
\begin{postulate}
    \label{pos:class}
    Any discrete classical transformation is described by a convex polyhedral cone. 
\end{postulate}

In general, any quantum Markov process is inscribed in a SDR cone. Following the
principles of quantum mechanics and the theorem (\ref{thm:scarrafone}) by Monràs and Winter 
it is natural to state the following postulate
\begin{postulate}
    \label{pos:sdrcone}
    Any discrete quantum mechanical transformation is described within a
    SDR cone.
\end{postulate}

To our knowledge quantum mechanics is a fundamental and universal theory \cite{CosmicBellTesHandst2017}.
Therefore, any classical dynamics should arise from a SDR cone, therefore the following
postulate arises naturally
\begin{postulate}
    Any classical (convex polyhedral) cone is embedded in a larger quantum (SDR) cone.
    \label{pos:embed}
\end{postulate}

Notwithstanding, this last postulate thus demands a mathematical treatment that
makes it more concrete. In Fig. \ref{fig:conos} we present a diagram of a quantum
SDR cone containing a classical cone.
The main problem then is how to reduce a SDR cone to a classical polyhedral one.
As mentioned before, a simple reduction mechanism is to suppose an instrument with a constant 
channel all the time. Then, because
of theorem (4.24) from \cite{WatrousTheTheoryof2018} the channel has at least one
fixed point. The output of the channel is reduced to a single state. 
However, a channel can have several fixed points, thus conforming a polyhedral
cone. This motivates the study of channels with multiple fixed points.
\section{Multiple fixed point channels}
We first cite a result from \cite{2209.06806v1} we have the following 
characterization of a channel in terms of its fixed point.

\begin{theorem}
    \label{thm:decomposition}
    Given a state $\sigma$ we can describe a trace-preserving separable family of channels with fixed point $\sigma$
    in terms of its Choi matrix $\mathcal{C}$ as follows
    \begin{equation}
        \mathcal{C}=\sigma\otimes\frac{(\ketbra{V_{max}}{V_{max}})^\intercal}{\lambda_{max}}
        +B\otimes(\mathds{I}-\frac{(\ketbra{V_{max}}{V_{max}})^\intercal}{\lambda_{max}}),
        \label{eq:chandecomp}
    \end{equation}
    $\lambda_{max}$ is the maximum eigenvalue of $\sigma$ and $\ket{V_{max}}$ its correspondent eigenvector.
    $B$ is a state. This description is valid for $\bra{V_{max}}B\ket{V_{max}}\leq\lambda_{max}$ and
    any input state $\bra{V_{max}}\rho\ket{V_{max}}\leq\lambda_{max}$.
%
\end{theorem}

Analogously to the SDP used in \cite{2209.06806v1} to find the above characterization
we have a way to find the channels that have as fixed points some desired
states.
\begin{corollary}
    The SDP for finding the channel with minimum trace with two fixed points is the
    following
\begin{equation}
    \begin{aligned}
        & \underset{X}{\text{maximize}}
        & & -\Tr[X] \\
        & \text{subject to}
        & &  \tr_{\mathcal{H}_2}[X(\mathds{1}_{\mathcal{H}_1}\otimes \sigma^{\intercal}_0)]=\sigma_0\\
        & & &  \tr_{\mathcal{H}_2}[X(\mathds{1}_{\mathcal{H}_1}\otimes \sigma^{\intercal}_1)]=\sigma_1\\
        & & & X\geq 0.
    \end{aligned}
    \label{eq:sdp}
\end{equation}

A trace-preserving channel is found as follows, in terms of a state $B$,
\begin{equation}
    \mathcal{C} = X + B\otimes(\mathds{1}-\tr_{H_1}[X]).
\end{equation}
We further require that 
\begin{equation}
    \tr[X(\mathds{1}_{H_1}\otimes B^\intercal)]< 1,
\end{equation}
so that iterations converge.
\end{corollary}
We nevertheless, lack a general characterization of the solution $X$ in terms
of the states $\sigma_0$ and $\sigma_1$ in contrast to the
one-state case. Observe that if $\{\sigma_0,\sigma_1\}\in\mathcal{H}$ then
$X\in\mathcal{H}\otimes\mathcal{H}$ and in general $X$ is a non-separable operator.
As a special case, we have the following characterization for when $X$ is separable.  

%

\begin{theorem}
    \label{thm:teoremaspecialChoi}
    A channel with Choi matrix $\mathcal{C}$ that has two fixed points $\sigma_0$ and
    $\sigma_1$ can be written as 
    \begin{align}
        \mathcal{C}&=\sigma_0\otimes\frac{\Pi_0^\intercal}{\tr[\Pi_0\sigma_0]}
        +\sigma_1\otimes\frac{\Pi_1^\intercal}{\tr[\Pi_1\sigma_1]}\nonumber\\
        &+B\otimes(\mathds{I}-\frac{\Pi_0^\intercal}{\tr[\Pi_0\sigma_0]}-\frac{\Pi_1^\intercal}{\tr[\Pi_1\sigma_1]}),
        \label{eq:chandecomp2}
    \end{align}
    if and only if the states $\sigma_0$ and $\sigma_1$ can be
    unambiguously discriminated. This means, there exist the positive semidefinite operators $\Pi_0$ and $\Pi_1$ 
    such that $\tr[\sigma_0\Pi_1]=\tr[\sigma_1\Pi_0]=0$
    and $\tr[\Pi_0\sigma_0]\neq0$ and $\tr[\Pi_1\sigma_1]\neq0$. 
    Also we require $\tr[B\Pi_0]/\tr[\Pi_0\sigma_0]+\tr[B\Pi_1]/\tr[\Pi_1\sigma_1]<1$.
\end{theorem}

\begin{proof}
    Observe that the channel $\mathcal{C}$ 
    fulfills the constrictions of the SDP (\ref{eq:sdp}), specifically,
    \begin{equation}
        \tr_{H_2}[\mathcal{C}(\mathds{1}\otimes\sigma_0)]= \sigma_0 + \sigma_1\frac{\tr[\sigma_0\Pi_1]}{\tr[\Pi_1\sigma_1]}-B\frac{\tr[\sigma_0\Pi_1]}{\tr[\Pi_1\sigma_1]},
    \end{equation}
    which is equal to $\sigma_0$ if and only if $\tr[\sigma_0\Pi_1]=0$. 
    Analogously $\tr_{H_2}[\mathcal{C}(\mathds{1}\otimes\sigma_1^\intercal)]= \sigma_1$
    if and only if $\tr[\sigma_1\Pi_0]=0$.
\end{proof}

\begin{figure}
    \includegraphics[scale=2]{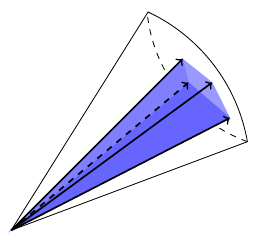}
    \caption{Cone embedding. The big transparent cone represents an SDR cone. The
        blue polyhedral cone represents a classical cone.}
    \label{fig:conos}
\end{figure}

The generalization to the cases of more fixed points is straightforward, we put
more conditions in the SDP (\ref{eq:sdp}), analogously add projectors
in the decomposition (\ref{eq:chandecomp2}).
\section{Simulation of a polyhedral cone}
Observe that a classical cone can be easily simulated with an SDR cone.
We start with a state $\rho_0$, apply a multiple-fixed point channel, end
up in a state $\sigma_i$. Then, we can apply a random channel to get
out of the fixed point, and apply the multiple-fixed point channel again
to end up in another state $\sigma_j$ etc, as depicted below.
\begin{align}
    \rho_0&\rightarrow\Phi^n[\rho_0]\approx\sigma_i\rightarrow\varphi_{\text rand}[\sigma_i]=\rho_0^\prime\\
    \rho_0^\prime&\rightarrow\Phi^n[\rho_0^\prime]\approx\sigma_j\rightarrow\varphi_{\text rand}[\sigma_j]=\rho_0^{\prime\prime}\\
    &\ldots
\end{align}
$\Phi$ is the multiple fixed point channel and $\varphi_{\text rand}$ is a random channel.
\section{Example}
Consider the Hilbert space spanned by two qubits. We can take the Bell basis and label
it as follows
\begin{align}
    \ket{V_0}&=\frac{\ket{00}+\ket{11}}{\sqrt{2}},\nonumber\\
    \ket{V_1}&=\frac{\ket{00}-\ket{11}}{\sqrt{2}},\nonumber\\
    \ket{V_2}&=\frac{\ket{01}+\ket{10}}{\sqrt{2}},\nonumber\\
    \ket{V_3}&=\frac{\ket{01}-\ket{10}}{\sqrt{2}}.
\end{align}
We also define the states
\begin{align}
    \ket{V_1^0}&=\alpha_0\ket{V_1}+\beta_0\ket{V_2},\nonumber\\
    \ket{V_2^0}&=\delta_0\ket{V_1}+\epsilon_0\ket{V_2},\nonumber\\
    \ket{V_1^1}&=\alpha_1\ket{V_1}+\beta_1\ket{V_2},\nonumber\\
    \ket{V_2^1}&=\delta_1\ket{V_1}+\epsilon_1\ket{V_2}.
\end{align}
Now let us define the mixed states
\begin{align}
    \sigma_0&=s_0\ketbra{V_0}{V_0}+s_1\ketbra{V_1^0}{V_1^0}+s_2\ketbra{V_2^0}{V_2^0}.\\
    \sigma_1&=r_0\ketbra{V_1}{V_1}+r_1\ketbra{V_1^1}{V_1^1}+r_2\ketbra{V_2^1}{V_2^1}.
\end{align}
We can easily verify that $\bra{V^1}\sigma_0\ket{V^1}=\bra{V^0}\sigma_1\ket{V^0}=0$ and therefore 
the characterization of fixed points from Eq. (\ref{eq:chandecomp2}) is valid.
\section{Discussion}
We develop here a general model to simulate a polyhedral cone using quantum systems. This
model is very general as it allows mixed states to be the vectors that subtend the cone.
We observe here a passing from a quantum system that has non-classical correlations
to a system that behaves classically, in terms of the stochastic processes they
produce. This is an example of a change of behavior from quantum dynamics into
classical dynamics which is closely related to the theory of einselection and
quantum Darwinism. Here the mechanism is analogous to the case of einselection
because it involves the repetitive action of a quantum channel.
However, the process yields the fixed points of the channel. We describe here
a way to construct quantum channels with desired specific fixed points. It can
be viewed as an engineering of quantum channels with multiple fixed points, which
gives rise to the classical behavior in terms of a polyhedral cone.

Observe that the most general characterization of our method is given by the SDP (\ref{eq:sdp})
which yields an operator $X$ that in general can be non-separable.
This means $X\in\mathcal{H}\otimes\mathcal{H}$ where $\mathcal{H}$ is the Hilbert space of
$\sigma_0$ and $\sigma_1$. However, in general, $X$ is not separable in those subspaces. In theorem
\ref{thm:teoremaspecialChoi} we explore the separable case. A full 
characterization of the solution $X$ in the non-separablenon-separable  case is a
matter of future research.

The mechanism that we study here is a specific one, however, the formalism could
be extended to consider other possible mechanisms that reduce the quantum
dynamics of a system into a classical one. This would extend the study of
quantum-to-classical transitions.
\section{Acknowledgements}
I thank fruitful discussions with A. Winter on this topic.
\bibliography{bibliography}
\end{document}